\documentclass[journal]{IEEEtran}
\usepackage{mathrsfs}
\usepackage{amsfonts}
\usepackage{amssymb}
\usepackage{dsfont}
\usepackage[dvips]{graphicx}
\usepackage{subfigure}
\usepackage{amsmath}

\usepackage{cite}

\usepackage{ctable}
\usepackage{multirow}

\newtheorem{theorem}{Theorem}
\newtheorem{lemma}{Lemma}
\newtheorem{corollary}{Corollary}

\begin{document}
\title{Bit Allocation Laws for Multi-Antenna Channel Feedback Quantization: Multi-User Case}
\author{Behrouz Khoshnevis and Wei Yu \thanks{The material in this paper has been presented
in part at the 44th Conference on Information Sciences and Systems, Princeton, NJ, Mar. 2010. This work was supported by the Natural Science and Engineering Research Council (NSERC) of Canada.}\thanks{The authors are with the Edward S. Rogers Sr. Department of Electrical and Computer Engineering, University of Toronto, 10 King's College Road, Toronto, Ontario, Canada M5S 3G4 (email: bkhoshnevis@comm.utoronto.ca; weiyu@comm.utoronto.ca).}} \maketitle

\newcommand{\qbar}{\bar{q}}
\newcommand{\gbar}{\bar{\gamma}}
\newcommand{\ds}{\mathds}
\newcommand{\md}{\textmd}
\newcommand{\be}{\begin{equation}}
\newcommand{\ee}{\end{equation}}
\newcommand{\ci}{\circ}
\newcommand{\eq}{\eqref}
\newcommand{\defined}{\stackrel{\md{def}}{=}}
\newcommand{\csi}{_{_{\md{CSI}}}}
\newcommand{\mbf}{\mathbf}
\newcommand{\mbb}{\mathbb}
\newcommand{\mc}{\mathcal}
\newcommand{\bh}{\mbf{h}}

\newcommand{\stupe}{\mbf{S}(\mbf{H})}
\newcommand{\bv}{\mbf{v}}

\newcommand{\D}{\mc{D}}
\newcommand{\M}{\mc{M}}
\newcommand{\NM}{\dot{N}}
\newcommand{\ND}{\ddot{N}}
\newcommand{\BM}{\dot{B}}
\newcommand{\BD}{\ddot{B}}
\newcommand{\NkM}{\dot{N}_k}
\newcommand{\NkD}{\ddot{N}_k}
\newcommand{\mcy}{\mathcal{\mathbf{y}}}
\newcommand{\bu}{\mbf{u}}
\newcommand{\bw}{\mbf{w}}

\newcommand{\BkD}{\ddot{B}_k}
\newcommand{\BkM}{\dot{B}_k}

\newcommand{\BbarM}{{\dot{{B}}_{\md{ave}}}}
\newcommand{\BbarD}{{\ddot{{B}}_{\md{ave}}}}

\newcommand{\sh}{\mc{S}(\mbf{h})}

\newcommand{\du}{\mathrm{d}}

\vspace{-1cm}

\begin{abstract}
This paper addresses the optimal design of limited-feedback downlink multi-user spatial multiplexing systems. A multiple-antenna base-station is assumed to serve multiple single-antenna users, who quantize and feed back their channel state information (CSI) through a shared rate-limited feedback channel. The optimization problem is cast in the form of minimizing the average transmission power at the base-station subject to users' target signal-to-interference-plus-noise ratios (SINR) and outage probability constraints. The goal is to derive the feedback bit allocations among the users and the corresponding channel magnitude and direction quantization codebooks in a high-resolution quantization regime. Toward this end, this paper develops an optimization framework using approximate analytical closed-form solutions, the accuracy of which is then verified by numerical results. The results show that, for channels in the real space, the number of channel direction quantization bits should be $(M-1)$ times the number of channel magnitude quantization bits, where $M$ is the number of base-station antennas. Moreover, users with higher requested quality-of-service (QoS), i.e. lower target outage probabilities, and higher requested downlink rates, i.e. higher target SINR's, should use larger shares of the feedback rate. It is also shown that, for the target QoS parameters to be feasible, the total feedback bandwidth should scale logarithmically with the geometric mean of the target SINR values and the geometric mean of the inverse target outage probabilities. In particular, the minimum required feedback rate is shown to increase if the users' target parameters deviate from the corresponding geometric means. Finally, the paper shows that, as the total number of feedback bits $B$ increases, the performance of the limited-feedback system approaches the perfect-CSI system as ${2^{-{B}/{M^2}}}$.

\end{abstract}

\section{Introduction}\label{S_Yek}

Multiple-input multiple-output (MIMO) technology can potentially provide significant performance improvements for wireless systems. More specifically, in the context of cellular communications, the availability of multiple antennas at the base-station allows it to simultaneously transmit to multiple users by multiplexing their data streams and hence improve the total downlink rate. Such systems are generally referred to as multi-user spatial multiplexing systems.

The performance of multi-user spatial multiplexing systems depends heavily on the amount of channel state information (CSI) at the base-station \cite{sharif05,jindal}. Such information is necessary to form the downlink transmission beams at the base-station and to perform rate/power adaptation for each user. Acquiring channel information, however, is a challenging issue in practice especially for frequency-division duplex (FDD) systems, where the uplink and downlink channels use different frequency bands. In these systems, the users need to explicitly quantize their channels and send back the quantized information through a shared feedback channel. The feedback link in such a system is usually a rate-limited control channel, hence the term \emph{limited-feedback systems}.


The scarcity of feedback bandwidth in limited-feedback systems necessitates efficient channel quantization codebook structures and an optimal allocation of feedback bits among those codebooks. This paper aims at deriving these structures and the corresponding bit allocation laws for limited-feedback multi-user spatial multiplexing systems.

\subsection{Related Work}

Multi-antenna communications with limited CSI is extensively studied in the literature for single-cell systems \cite{jindal, yoo, huang, caire, sharif05, ding07, kount06, icc, ciss, huang09} and to some extent for cooperative multi-cell networks \cite{bhagavatula10, bhagavatula10_2, Han09}. This paper focuses on single-cell multi-user systems.


The major advantage of multi-antenna multi-user systems with respect to single-user systems is the sum-rate \emph{multiplexing gain}, which follows from the fact that multiple simultaneous transmissions can be established in a multi-antenna downlink. In order to preserve this gain in limited-feedback systems, the author of \cite{jindal} shows that the total feedback rate should scale linearly with the signal-to-noise ratio (SNR) in dB scale. The work in \cite{jindal} addresses a setup with small number of users. In a network with large number of users, there is another source of sum-rate improvement, referred to as \emph{multi-user diversity gain}, which is realized by the base-station opportunistically scheduling users with favorable channel conditions. For scheduling, a well justified approach is to choose users with high channel gains and near-orthogonal channel directions \cite{yoo,huang,caire,sharif05,ding07,huang09,kount06}. The authors of \cite{yoo} specifically show that one needs the channel gain information (CGI) in addition to the quantized channel direction information (CDI) in order to realize the multi-user diversity gain. The gain information however is assumed to be perfect in \cite{yoo}. The split of feedback bits between CGI and CDI quantization introduces an interesting tradeoff between multiplexing gain and diversity gain. This tradeoff is studied by \cite{kount06}, where the authors numerically show that more bits should be used for CGI quantization in order to benefit from multi-user diversity gain as the number of users increases.

For the purpose of precoding the information intended to scheduled users, two distinct approaches are proposed in the literature. The first group of work uses zero-forcing beamforming based on the quantized directions \cite{yoo,ding07}. The second group adopts a codebook (or multiple codebooks) of orthonormal beamforming vectors and selects beamforming directions based on the signal-to-interference-plus-noise ratio (SINR) feedbacks from the users \cite{huang,sharif05,huang09}. The authors of \cite{huang09} specifically claim that in the regime of large number of users, the orthogonal beamforming approach outperforms the zero-forcing method. Finally, the authors of \cite{caire} present a more thorough analysis of multi-user limited-feedback systems considering joint training, scheduling, and beamforming.

\subsection{System Model}\label{system_model}

This paper considers a limited-feedback multi-user system with a $M$-antenna base-station and $M$ single-antenna users. The users send their quantized CSI to the base-station through a feedback link with a capacity of $B$ bits per each downlink transmission block. Based on this quantized information, the base-station then comes up with the downlink transmission powers and beamforming vectors.

Our goal is to optimize the system performance subject to the total feedback rate constraint. Although the proposed approach is rather generic, we impose several simplifying assumptions on the system model in order to achieve closed-form solutions. These assumptions are listed below and will be explicitly mentioned and justified whenever needed throughout the text:

\noindent \emph{Assumptions}:
\begin{itemize}
\item[A1.] Most of the analysis in this paper is based on the assumption of \emph{high resolution quantization}, i.e. $B\rightarrow\infty$. This assumption is made mainly to achieve a tractable formulation for system optimization and has been used frequently in the literature of quantization theory \cite{gersho,zheng}. In addition, the paper investigates, both numerically and analytically, the system performance with moderate values of $B$, and examines the regime where the high resolution results are applicable.
\item[A2.] Users' channels are independent and identically distributed (i.i.d.).
\item[A3.] Users' channel directions are uniformly distributed over the $M$-dimensional unit hypersphere. This includes the Rayleigh i.i.d. channels as an special case \cite{heath03}.
\item[A4.] Users' channel magnitudes are independent from channel directions and can have an arbitrary distribution.
\item[A5.] We assume a product structure for the channel quantization codebook, i.e. channel magnitudes and channel directions are quantized independently. This product structure is justified in \cite{globecom,icc}. According to the high resolution assumption, if we denote the magnitude and direction codebook sizes by $\NM_k$ and $\ND_k$, we have $\NkM,\NkD\rightarrow\infty$ for each user $1\leq k \leq M$.
\item[A6.] The beamforming vectors are zero-forcing directions for the quantized directions. This assumption is to mimic the perfect-CSI case, where zero-forcing beamforming is shown to achieve asymptotically optimal sum-rate scaling with SNR \cite{jindal05,yoo06}.
\item[A7.] The user channels are assumed to be real vectors. Although this assumption appears in the earlier literature, e.g. \cite{geyi,pun}, we use it mainly for the ease of geometric representation of the quantization regions and the corresponding calculations. The extension of the analysis to complex space is discussed in Section \ref{S_Haft}.
\end{itemize}

\subsection{Problem Formulation}

We formulate the system design problem as the minimization of the average sum power subject to the users' outage probability constraints. In order to differentiate between the users' QoS requirements, we assume different target SINR's and outage probabilities across the users. Our goal is to derive the optimal split of feedback bits among the users and the corresponding magnitude and direction quantization codebooks.

Different variations of the power minimization formulation are used in the literature, e.g. in \cite{schubert04,rashid98,wiesel06,sharma06,tarighat05}, as an appropriate formulation for fixed-rate delay-sensitive applications, e.g. voice over IP, video conferencing, and interactive gaming. The reliability of the fixed-rate link is usually achieved by applying power control at the base-station to compensate for the channel fading, as in Wideband Code Division Multiple Access (WCDMA) system standards \cite{book_wcdma}.

An alternative problem formulation is to maximize the average sum rate subject to a power constraint \cite{jindal, yoo, huang, sharif05, ding07, kount06,huang09}. In most of these formulations, the transmission power is fixed and the quantized information is used only to adapt the transmission rates. This type of formulation is more appropriate for variable-rate communication systems, e.g. Worldwide Interoperability for Microwave Access (WiMAX) and 3GPP Long Term Evolution (LTE) system standards \cite{wimax,lte}.

We make the following comments on the existing rate-maximization formulations in the literature:

\begin{itemize}
\item[1.] The sum-rate maximization problem assumes equal priority for all users in terms of their applications. Since the users are not differentiated based on their QoS measures, this formulation cannot answer the question of how to optimally split the feedback bits among users with different QoS requirements.
\item[2.] The channel gain information (CGI) is clearly an important factor in scheduling the users and also in setting the downlink rate for each user. However, most of the existing literature either assumes perfect channel gain information or completely ignores this information. It is therefore not clear how to optimally split the feedback bits between channel direction and channel gain quantizers in the context of sum-rate maximization problem. The works that look into this problem are mainly numerical and lack closed-form solutions \cite{kount06,huang09}.
\end{itemize}

One solution to the first issue raised above is to prioritize users with certain weights and consider the weighted sum-rate maximization instead of the sum-rate itself. These weights can be set by the scheduler based on users' QoS requirements. The proportional fairness scheduler for example sets the weights based on users' backlogged traffic by assigning a higher weight to a user with larger backlog. This type of formulation appears for example in \cite{schubert,girici,bjornson,yu10} for perfect-CSI systems. Generalizing this formulation to limited-feedback systems however appears to be a difficult problem. We are not specifically aware of any work that addresses such a problem. For this reason, this paper resorts to a power minimization formulation, which can easily incorporate the QoS constraints by assuming different target SINRs and outage probabilities across the users. The power minimization formulation also simplifies the CGI/CDI bit allocation problem and provides insight to the system design by allowing for closed-form asymptotic bit allocation solutions.


\subsection{Proposed Approach}
Our approach for solving the power minimization problem is to fix the channel outage regions in advance and transform the problem to a robust optimization problem. Assuming zero-forcing beamforming vectors, we first formulate the robust power control problem in the form of a semi-definite programming (SDP) problem. Using an approximate upper bound solution to the SDP problem, we then derive the codebook structures and the corresponding bit allocations in the asymptotic regime where $B\rightarrow\infty$. Within the proposed approximate optimization framework, we show that for channels in real space:

\begin{itemize}
\item[1.] The optimal number of channel direction quantization bits is $M-1$ times the number of channel magnitude quantization bits, where $M$ is the number of base-station antennas.

\item[2.] The share of the $k$th user from the total feedback rate is controlled by $\log{\gamma_k}$ and $\log{1/q_k}$, where $\gamma_k$ and $q_k$ are the user's target SINR and outage probability. As a general rule, a user with a lower target outage probability and higher target SINR needs a higher channel quantization resolution and therefore requires a larger share of the total feedback rate.

\item[3.] For the outage probability constraints to be feasible, the total feedback rate should scale logarithmically with $\gbar$, the geometric mean of the target SINR values, and $1/\qbar$, the geometric mean of the inverse target outage probabilities. Moreover, the minimum required feedback rate increases if the users' target parameters deviate from the average parameters $\gbar$ and $\qbar$, i.e. there is a feedback rate penalty for serving users with non-similar target parameters. The higher the deviation, the higher the penalty.

\item[4.] As the total feedback rate $B$ increases, the performance of the limited CSI system approaches the performance of the perfect-CSI system as $2^{-\frac{B}{M^2}}$.
\end{itemize}
These optimality results are based on minimizing an upper bound of the sum power. The closeness of the upper bound to the exact sum power is verified numerically in the paper.

\subsection{Organization of the Paper}

The remainder of this paper is organized as follows. Section \ref{S_Do} provides an overview of the perfect CSI system. Section \ref{S_Se} presents the system design problem in its general form and describes our approach in transforming it to a robust design problem. Section \ref{S_Panj} describes the product channel quantization codebook structure. In Section \ref{S_Char}, we study the power control optimization for fixed quantization codebooks and derive an upper bound for the average sum power. By using the sum power upper bound, we then optimize the product codebook structures in Section \ref{S_Shish} and derive the asymptotically optimal bit allocation laws. Finally, Section \ref{S_Hasht} presents the numerical results and Section \ref{S_Haft} concludes the paper.

\emph{Notations:} Most of the computations in this paper are in real space. The logarithm functions are base 2. The angle between any two unit-norm vectors $\bu$ and $\bv$ is defined as $\angle(\bu,\bv)=\arccos |\bu^T\bv|$ so that $0\leq \angle(\bu,\bv)\leq\pi/2$.

\section{Multiuser Spatial Multiplexing System with Perfect CSI: Outage is Inevitable} \label{S_Do}

We start by assuming perfect CSI at the base-station and show that, unlike a single-user system, outage is inevitable in the multi-user system even with perfect CSI. The difference with the single-user case is due to the fact that the base-station in a multi-user system needs to distinguish the users spatially and when the user channels are closely aligned it is not possible to satisfy the users' target SINR's with a bounded average transmission power.

Consider a multi-user downlink channel with $M$ antennas at the base-station and $M$ users each with a single antenna. Let $\bh_k\in\mathds{R}^M$, $\mbf{v}_k\in\ds{R}^M$, $P_k$, and $\gamma_k$ denote respectively, the user channel, the unit-norm beamforming vector, the allocated power, and the target SINR for the $k$th user, $1\leq k\leq M$. The minimization of the transmission sum power subject the user SINR constraints is formulated as follows:
\begin{align}\label{CSI-sumpower}
&\min\limits_{P_k,\mathbf{v}_k} \sum_{k=1}^{M}{P_k}\\
&\textmd{~~s.t.~}
\frac{P_k\left|\mathbf{h}_k^{T}\mathbf{v}_k\right|^2}{\sum\limits_{l\neq
k}{P_l\left|\mathbf{h}_k^{T}\mathbf{v}_l\right|^2}+1}\geq
\gamma_k,~~k=1,2,\cdots,M\nonumber
\end{align}
where the receiver noise power is assumed to be $1$ for all users.

A suboptimal solution for problem \eqref{CSI-sumpower} is to use zero-forcing (ZF) beamforming vectors $\mbf{v}_k$ to eliminate the interference and find the power values $P_k$ that satisfy the constraints with equality. This solution is asymptotically optimal in the high SNR regime \cite{jindal05,yoo06}. Clearly, for the zero-forcing solution to be applicable, the users' vector channels need to be linearly independent. Since the channels are assumed to be independent random vectors, this condition is satisfied almost surely, i.e., with probability one.

An important matter to consider with this solution is that the transmission powers would need to be extremely high when the users' channels are closely aligned, as the ZF beamforming vectors would be almost perpendicular to the corresponding channels in such cases. Therefore, it is not possible to always satisfy the SINR constraints with a bounded average power and as a result, a certain degree of outage must be tolerated by the users.

To see this rigorously, define $\theta_k=\angle(\mathbf{h}_k,\mathbf{H}_{-k})$, where $0\leq\theta_k\leq\frac{\pi}{2}$, and $\mathbf{H}_{-k}=\textmd{span}(\{\mathbf{h}_l|l\neq k\})$. For zero-forcing beamforming vectors $\bv_k$ we have $\angle(\bv_k,\bh_k){=}\frac{\pi}{2}{-}\theta_k$. Assume that the users' channels are i.i.d. with uniformly distributed directions and independent channel magnitudes (with arbitrary distributions). The average sum power of the zero-forcing method is given by
\begin{align}
P_{_{\md{MU,CSI}}}{=}\sum_{k{=}1}^M{\mathbb{E}\left[\frac{\gamma_k}{\left|\mathbf{h}_k^{T}\mathbf{v}_k\right|^2}\right]}
{=} \sum_{k{=}1}^{M}{\gamma_k} \mathbb{E}\left[\frac{1}{\|\mathbf{h}_k\|^2}\right]\mathbb{E}\left[\frac{1}{\sin^2(\theta_k)}\right]. \nonumber
\end{align}
As $\theta_k$ is uniformly distributed in $[0,\frac{\pi}{2}]$, the expectation
of $1/\sin^2(\theta_k)$ becomes unbounded.

To avoid unbounded transmit power, the users should tolerate certain degrees of outage. A reasonable approach is to declare outage for user $k$, i.e. set $P_k=0$, when \[0\leq\theta_k<\theta_k^{\circ},\] where $\theta_k^{\circ}\ll 1$ is the smallest acceptable angle between $\mbf{h}_k$ and $\mbf{H}_{-k}$. With this assumption the average sum power is given by
\begin{align}\label{P-CSI}
P_{_{\md{MU,CSI}}}&=\sum_{k=1}^{M}{\gamma_k\mathbb{E}{\left[{1}/{\|\mathbf{h}_k\|^2}\right]}
\frac{1}{\pi/2}\int_{\theta_k^\circ}^{\pi/2} \frac{\md{d}\theta_k}{\sin^2\theta_k}}\nonumber\\
&=\frac{2}{\pi}\sum_{k=1}^{M}{\gamma_k
\mathbb{E}{\left[{1}/{\|\mathbf{h}_k\|^2}\right]}\cot \theta_k^{\ci}}\approx
\frac{2\rho_{_{\md{MU,CSI}}}}{\pi}\sum_{k=1}^{M}{\frac{\gamma_k}{\theta_k^{\ci}}},
\end{align}
where the approximation holds for $\theta_k^{\circ}\ll 1$ and \be{\label{rho-CSI}}\rho_{_{\md{MU,CSI}}}\defined\mathbb{E}{\left[{1}/{\|\mathbf{h}_k\|^2}\right]}\ee for i.i.d. users. The corresponding outage probabilities are $p_{\textmd{out},k}={2\theta_k^{\ci}}/{\pi}\nonumber$ for $1\leq k\leq M$.

Having studied the perfect CSI system, the next section describes a general framework for the limited-feedback system design. The insights achieved by studying this general form are used in the later sections for system design and optimization with product channel quantization codebooks.

\section{System Design Problem and Vector Channel Quantization: General Form} \label{S_Se}
To clarify the arguments, we start by some basic definitions. By a \emph{vector channel quantization codebook} $\mc{C}$ of size $N$, we mean a partition of $\ds{R}^M$ into $N$ disjoint \emph{quantization regions} $S^{(n)}$, $1{\leq} n{\leq} N$:\[\mc{C}{=}\{S^{(1)},S^{(2)},\cdots,S^{(N)}\}.\] For every quantization codebook $\mc{C}$, we also define a \emph{quantization function} \[\mc{S}(\mbf{h}):\ds{R}^M\rightarrow\mc{C},\] which returns the quantization region that $\mbf{h}\in\ds{R}^M$ belongs to.

Now, for each user $1{\leq} k{\leq} M$, associate a codebook $\mc{C}_k$ of size $N_k$ and the corresponding quantization function $\mc{S}_k(\mbf{h}_k)$, where $\mbf{h}_k$ is the $k$th user's channel. Further, define the ordered $M$-tuples
\begin{align}
&\mbf{H}\defined[\bh_1^T,\bh_2^T,\cdots,\bh_M^T]\in \ds{R}^{M^2},\nonumber\\
&{\mbf{S}}(\mbf{H})\defined[\mc{S}_1(\bh_1),\mc{S}_2(\bh_2),\cdots,\mc{S}_M(\bh_M)]\in \prod_{k=1}^{M}{\mc{C}_k}.\nonumber
\end{align}

For a given total number of quantization (feedback) bits $B$, target SINR values $\gamma_k$, and target outage probabilities $q_k$, the system design problem is formulated as follows:
\begin{align}\label{general-opt}
&\min\limits_{\substack{\mc{C}_k,N_k,\\P_k(\stupe),\\\bv_k(\stupe)}} \mathbb{E}_{\mbf{H}}\left[\sum_{k=1}^{M}{P_k(\stupe)}\right]\\
&~~~~ \textmd{s.t.} \quad \prod_{k=1}^{M}N_k= 2^B,\nonumber\\
&\qquad\qquad \md{prob}\left[\frac{P_k(\stupe)\left|\mathbf{h}_k^{T}\mathbf{v}_k(\stupe)\right|^2}{\sum\limits_{l\neq
k}{P_l(\stupe)\left|\mathbf{h}_k^{T}\mathbf{v}_l(\stupe)\right|^2}{+}1} {<} \gamma_k\right] {\leq}q_k,\nonumber \nonumber\\&\hspace{2.5in}k=1,2,\cdots,M\nonumber
\end{align}
where the optimization is over the quantization codebooks $\mc{C}_k$, codebook sizes $N_k$, the power control functions $P_k(\stupe):\prod_{k=1}^{M}{\mc{C}_k} \rightarrow \ds{R}_{+}$, and the beamforming functions $\bv_k(\stupe):\prod_{k=1}^{M}{\mc{C}_k}\rightarrow \mathfrak{U}_M$, where $\mathfrak{U}_M$ is the unit hypersphere in $\ds{R}^M$.

An exact solution to this problem is intractable. Our approach in simplifying the problem is to fix the outage scenarios in advance and transform the design problem to a robust design problem that guarantees the target SINR's for the no-outage scenarios.

Define the \emph{outage region} $\Omega_k\subset\prod_{k}{\mc{C}_k}$ for user $k$ such that $\md{prob}[\stupe\in\Omega_k]=q_k$. Also define $I_k(\stupe)$ as the \emph{activity flag} for user $k$: \[I_k(\stupe)=\mc{I}(\stupe\in\Omega_k^c),\] where $\mc{I}(\cdot)$ is the logic true function. Whenever a user's channel resides outside the user's predefined outage region, the activity flag is on and the user must be served by the base-station, i.e. the user should not face an outage.

Let us fix the codebook sizes $N_k$ for now. For a robust system design, we need to design the codebooks, the power control functions, and the beamforming functions such that the target SINR's are guaranteed whenever $I_k(\stupe){=}1$:
\begin{align}\label{robust-design}
&\min\limits_{\substack{\mc{C}_k,\\P_k(\stupe),\\\bv_k(\stupe)}} \mathbb{E}_{\mbf{H}}\left[\sum_{k=1}^{M}{P_k(\stupe)}\right]\\
&\textmd{s.t.} {\inf_{\bw{\in} \mc{S}_k(\bh_k)}}\frac{P_k(\stupe)\left|\mathbf{w}^{T}\mathbf{v}_k(\stupe)\right|^2}{\sum\limits_{l\neq
k}{P_l(\stupe)\left|\mathbf{w}^{T}\mathbf{v}_l(\stupe)\right|^2}{+}1} \geq \gamma_kI_k(\stupe), \nonumber\\& \hspace{1.in} \forall~ \mbf{H}\in\ds{R}^{M^2}~\md{and}~k=1,2,\cdots,M\label{robust_constraint}
\end{align}
Note that by including the activity flag in the constraint \eqref{robust_constraint}, this formulation guarantees the target SINR when $I_k=1$  and returns $P_k{=}0$ when $I_k{=}0$. Also note that the activity flags are fixed in advance such that $\md{prob}[I_k(\stupe){=}0]{=}q_k$.


The design problem in \eqref{robust-design} is a complicated problem. In order to achieve a tractable reformulation, we accept two main simplifying assumptions (Assumptions A5 and A6 in Section \ref{system_model}):
\begin{itemize}
\item We assume a product structure for the channel quantization codebook, where the channel magnitude and the channel direction are quantized independently. Such a product structure, also known as shape-gain quantization in the literature \cite{gersho}, provides several practical advantages including faster quantization and lower storage requirement for the quantization codebooks. Product codebook structures are also shown to be sufficient structures for effective interference management in multi-user systems \cite{icc}.
\item We assume a fixed zero-forcing design for the following reason. Let us fix the quantization codebooks $\mc{C}_k$ and consider a time snapshot of the optimization problem in \eqref{robust-design} with a fixed channel realization $\mbf{H}$ and fixed quantized information $\stupe$. For this specific point in time, we intend to optimize the beamforming vectors $\bv_k$ and power levels $P_k$ such that the worst-case SINR conditions in \eqref{robust_constraint} are satisfied. It can be shown that this problem is non-convex with respect to $P_k$ and $\bv_k$ and therefore difficult to solve in general. However, if one fixes the vectors $\bv_k$, the problem would be convex in terms of $P_k$, since the objective and constraints are linear in $P_k$. Similar observations are made in \cite{icc3,icc5,icc7}. In order to preserve the convexity structure, we therefore assume that the beamforming vectors $\mbf{v}_k(\stupe)$ are fixed as zero-forcing vectors for the quantized channel directions. This is to mimic the zero-forcing beamforming with perfect CSI, where the base-station knows the exact channel directions. The exact definition of the quantized channel directions is provided later in Section \ref{S_Panj-A-2}.

\end{itemize}

With these simplifying assumptions, the robust design problem reduces to the following subproblems: 1) optimizing the power control function for fixed beamforming vectors and codebook structures; 2) optimizing the product codebook structure itself. The following sections address these subproblems.


\section{Product Quantization Codebook Structure} \label{S_Panj}
In this section, we describe the product quantization codebook structures and specify the corresponding outage regions. To be more exact, for a given target outage probability $q_k$, we specify the magnitude and direction outage regions such that \be\label{tot-outage} q_k=\dot{q}_k+\ddot{q}_k, \nonumber\ee
where the \emph{magnitude outage probability} $\dot{q}_k$ is the probability that the channel magnitude resides in the specified \emph{magnitude outage region} and the \emph{direction outage probability} $\ddot{q}_k$ is the probability that the channel direction resides in the specified \emph{direction outage region}.

\subsection{Magnitude Quantization Codebook and Magnitude Outage Region} \label{S_Panj-A-1}
For each user $1\leq k\leq M$, we use a magnitude quantization codebook \[\mbb{Y}_k=\left\{y_{k}^{(1)},y_k^{(2)},\cdots,y_k^{(\NkM)}\right\}\]
for quantizing the channel magnitude squared $Y_k\defined\|\bh_k\|^2$. Here $y_k^{(n)}$ are the quantization levels and $\NkM$ is the magnitude codebook size.

For a given \emph{magnitude outage probability} $\dot{q}$, we define the \emph{magnitude outage region} as the leftmost quantization interval $\big[0,y_k^{(1)}\big)$. The first quantization level is therefore fixed as
\be\label{y1} y_{k}^{(1)}=F^{-1}(\dot{q}_{k}),\ee
where $F^{-1}(\cdot)$ is the inverse cumulative distribution function (cdf) of $Y_k\defined\|\bh_k\|^2$.

We further define $\dot{\mc{C}}_{k}$ as the set of magnitude quantization regions, i.e. the set of quantization intervals for $\|\bh_k\|=\sqrt{Y_k}$:
\be\label{M-code}\dot{\mc{C}}_{k}=\left\{J_k^{(1)},J_k^{(2)},\cdots,J_k^{(\NkM)}\right\},\ee
where $J_k^{(n)}=\left[\left.\sqrt{y_{k}^{(n)}},\sqrt{y_{k}^{(n+1)}}\right)\right.$ and $y_k^{(\NkM+1)}\defined\infty$. Note that the definition uses the square root of the levels as the quantization levels $y_{k}^{(n)}$ are defined for quantizing $\|\bh_k\|^2$.

Finally, for $Y_k\geq y_{k}^{(1)}$, we define the \emph{quantized magnitude} $\tilde{Y}_k$ as the quantization level in $\mbb{Y}_k$ that is in the immediate left of $Y_k$, i.e.,
\be\label{quant-mag} \tilde{Y}_k=y_k^{(n)} ~~~~ \md{if} ~~~~ y_k^{(n)}\leq Y_k <y_k^{(n+1)}.\ee

For reasons that are clarified later in Section \ref{S_Shish}, we are interested in a magnitude quantization codebook that minimizes $\mathbb{E}\left[{1}/{\tilde{Y}_k}\right]$. It is shown in \cite{PartI} that the optimal codebook with such a criterion is uniform (in dB scale) in the asymptotic regime where $\NkM\rightarrow\infty$. We denote such an optimal codebook by $\mbb{Y}^\star_k$ and refer to it as the \emph{uniform magnitude quantization codebook} in the remainder of this paper.

The work in \cite{PartI} further shows that the uniform codebook $\mbb{Y}^\star_k$ satisfies the following upper bound:
\be\label{Pavescale-mu}
\mathbb{E}\left[\frac{1}{\tilde{Y}^\star_k}\right]<\rho_{_{\md{MU,CSI}}} \left(1+\NkM^{-\zeta_k\left(\NkM\right)}+\omega\NkM^{-2\zeta_k\left(\NkM\right)}\right),
\ee
where \[\rho_{_{\md{MU,CSI}}}=\mbb{E}[{1}/{Y_k}]=\mbb{E}[{1}/{\|\bh_k\|^2}]\]
as defined in \eqref{rho-CSI}. On the left hand side of $\eqref{Pavescale-mu}$, the variable $\tilde{Y}^\star_k$ is the quantized magnitude variable associated with the uniform codebook $\mbb{Y}^\star_k$. On the right-hand side of \eqref{Pavescale-mu}, \[\omega\defined \frac{\mbb{E}[Y_k]}{\eta^2\mbb{E}[{1}/{Y_k}]},\] where \[\eta=\lim_{y\rightarrow\infty}{{-f(y)}/{f'(y)}}\] and $f(\cdot)$ is the probability density function (pdf) of $Y_k$. The function $\zeta_k(n)$ depends on the magnitude outage probability $\dot{q}_k$ and is defined as the solution to the following equation:
\[n^{-\zeta_k(n)}\left(1+n^{-\zeta_k(n)}\right)^{n-1}=\frac{\eta}{y_{k}^{(1)}},\]
where $y_{k}^{(1)}=F^{-1}(\dot{q}_k)$ as defined earlier. It can be shown that for any $\dot{q}_k>0$ and hence $y_{k}^{(1)}>0$, we have
\be\label{zeta-limit} \lim_{n\rightarrow\infty}{\zeta_k(n)}=1.\ee

The bound in \eqref{Pavescale-mu} and the limit in \eqref{zeta-limit} are used in Section \ref{S_Shish} for optimization of the product channel quantization codebook structure.

\subsection{Direction Quantization Codebook and Direction Outage Region} \label{S_Panj-A-2}
For each user $1{\leq} k{\leq} M$, we use a Grassmannian codebook $\mbb{U}_k$ of size $\ND_k$ for direction quantization: \be\label{grass-code-mu} \mbb{U}_k=\left\{\bu_k^{(1)},\bu_k^{(2)},\cdots,\bu_k^{(\ND_k)}\right\},\ee where $\bu_k^{(n)}$ vectors are $M$-dimensional unit-norm Grassmannian codewords.

Every channel realization $\bh_k$ is mapped to a vector $\tilde{\bu}_k(\bh_k)\in\mbb{U}_k$ that has the smallest angle with $\bh_k$:
\be\label{dir-quant-vec-mu} \tilde{\bu}_k(\bh_k)=\arg\min_{\bu\in\mbb{U}_k} \angle{(\bh_k,\bu)}.\ee
The vector $\tilde{\bu}_k$ is referred to as the \emph{quantized direction} for the channel realization $\bh_k$. The corresponding quantization regions, according to the Gilbert-Varshamov argument \cite{Barg}, can be covered by the following spherical caps:
\be\label{dir-quant-regions-mu} \ddot{\mc{C}}_k=\left\{B_k^{(1)},B_k^{(2)},\cdots,B_k^{(\ND_k)} \right\}, \ee
where \[B_k^{(n)}=\left\{\mbf{w}\in\mathfrak{U}_M\left|\angle{(\mbf{w},\bu_k^{(n)})}<\phi_k\right.\right\}\] is the spherical cap around $\bu_k^{(n)}$ as shown in Fig. \ref{Fig0}. In this definition, \[\phi_k=\arcsin{\delta_k}\] is the angular opening of the caps and $\delta_k$ is the minimum chordal distance of $\mbb{U}_k$. It should be noted that covering the direction quantization regions with the spherical caps enlarges the quantization regions. By considering the constraint \eqref{robust_constraint} in the robust design problem \eqref{robust-design}, such enlargement of the regions will lead to an upper bound for the average transmission power.

\begin{figure}
\centering
\includegraphics[width=1in]{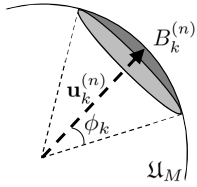}
\caption{Spherical cap $B_k^{(n)}$ around $\bu_k^{(n)}$ on the unit hypersphere $\mathfrak{U}_M$. }
\label{Fig0}
\end{figure}

In order to describe the \emph{direction outage regions}, define \[\theta_k=\angle(\tilde{\mathbf{u}}_k,\tilde{\mathbf{U}}_{-k}),\] where \[\tilde{\mathbf{U}}_{-k}=\textmd{span}(\{\tilde{\bu}_l|l\neq k\})\] and $\tilde{\bu}_k$ is the quantized direction for user $k$. This is a similar definition as in Section \ref{S_Do}, except that the exact channels $\mbf{h}_k$ are replaced with the quantized directions $\tilde{\mbf{u}}_k$. Similar to the discussion in Section \ref{S_Do} for the perfect CSI system, we say that user $k$ is in \emph{direction outage} if \[0\leq \theta_k<\theta_k^\circ,\] where $\theta_k^\circ$ is the minimum acceptable angle between $\tilde{\bu}_k$ and ${\mbf{U}}_{-k}$. This implicitly defines the direction outage regions of the users. Finally, assuming that $\theta_k$ is approximately uniform in $[0,\pi/2]$\footnote{This holds if the channel direction are uniformly distributed (Assumption A3 in Section \ref{system_model}) and the codebooks $\mbb{U}_k$ undergo sufficient random rotations.}, the \emph{direction outage probability} is given by  \be\label{dir-out-prob-mu} \ddot{q}_k\approx\frac{2}{\pi}\theta_k^\circ.\ee

In Section \ref{S_Shish}, which addresses the product codebook optimization, we will need the following inequality, which describes the dependence between the angular opening $\phi_k$ and the direction codebook size $\NkD$:
\be\label{phi-bound-mu} \phi_k\approx\sin\phi_k<4\lambda_M\NkD^{-\frac{1}{M-1}},\ee
where $\lambda_M{=}\left({\sqrt{\pi}\Gamma((M{+}1)/2)}/{\Gamma({M/2})}\right)^{\frac{1}{M{-}1}}$. This inequality holds for large enough values of $\NkD$ and its proof is presented in \cite{PartI}. The approximation on the left-hand side of \eqref{phi-bound-mu} assumes $\phi\ll1$, which is justified by the high resolution assumption $\NkD\gg1$ (Assumption A1 and A5 in Section \ref{system_model}).

\subsection{Product Codebook Structure} \label{S_Panj-A-3}
Using our definitions of the magnitude and direction quantization regions, we define product channel quantization codebook $\mc{C}_k$ for each user $k$ as follows:
\be\label{product-code-mu} \mc{C}_k=(\dot{\mc{C}}_k\times\ddot{\mc{C}}_k) \cup \mc{O}_k, \ee
where $\dot{\mc{C}}$ and $\ddot{\mc{C}}$ are the magnitude and direction quantization regions in \eqref{M-code} and \eqref{dir-quant-regions-mu}, and \[\mc{O}_k=\left\{\bh\left|\|\bh\| < \sqrt{y_k^{(1)}}\right.\right\}\]
is a ball centered at origin corresponding to the magnitude outage region.

Based on the definitions of the magnitude and direction outage regions, the activity flag for user $k$ is given by  \be\label{flag2} I_k=\mc{I}\left(\theta_k\geq\theta_k^\circ ~\wedge~ \|\bh\|\geq\sqrt{y_k^{(1)}}\right),\ee where $\mc{I}(\cdot)$ is the logic true function, and $\theta_k$, $\theta_k^\circ$, and $y_k^{(1)}$ are defined in Sections \ref{S_Panj-A-1} and \ref{S_Panj-A-2}.

According to the activity flag expression in \eqref{flag2}, the outage event for user $k$ can be expressed as the union of the magnitude outage event, corresponding to the channel magnitude lying in the magnitude outage region, and the direction outage event, corresponding to the channel direction lying in the direction outage region. By using the union probability formula we have
\be
\md{prob}[I_k=0] \leq \dot{q}_k+\ddot{q}_k \approx F(y_{k}^{(1)})+\frac{2}{\pi}\theta_k^\circ, \nonumber
\ee
where $\dot{q}_k$ and $\ddot{q}_k$ are the magnitude and direction outage probabilities respectively and the approximation follows from \eqref{dir-out-prob-mu}. In order to satisfy the target outage probability $q_k$, we therefore impose the following constraint\footnote{Note that this constraint is stronger than $\md{prob}[I_k=0]\leq q_k$ and would therefore lead to an upper bound on the objective function (average sum power). This agrees with the direction of our analysis in achieving an upper bound for the sum power.}:
\be\label{total-outage}
F(y_{k}^{(1)})+\frac{2}{\pi}\theta_k^\circ \leq q_k.
\ee

Finally, the product channel quantization codebook size is given by \be\label{code-size} |\mc{C}_k|=N_k=\NkM\NkD+1,\ee
where $\NkM$ and $\NkD$ are the magnitude and direction codebook sizes respectively.

To summarize, for a given target outage probability $q_k$ and a given codebook size $N_k$, we propose a product quantization structure and specify its outage regions and the corresponding activity flags such that $\md{prob}[I_k=0]=q_k$. The proposed product codebook structure is parameterized by the magnitude and direction codebook sizes $\NkM$ and $\NkD$, the minimum acceptable channel magnitude $y_k^{(1)}$, and the minimum acceptable directional separation $\theta_k^\circ$.

\section{Optimization of the Power Control Function with Sector-Type Channel Uncertainty Regions} \label{S_Char}
Assuming the product codebook structure in Section \ref{S_Panj}, this section addresses the optimization of the power control function. For this purpose, we fix the quantization codebooks $\mc{C}_k$ and the corresponding outage regions. Furthermore, as mentioned in Section \ref{S_Se}, we make the simplifying assumption that the beamforming vectors are the zero-forcing vectors for the quantized directions.

For the product quantization codebooks considered in this paper, the quantization (or channel uncertainty) regions are sector-type regions as shown in Fig. \ref{Fig1}. A sector-type region is parameterized as \[S(R,r,\tilde{\mbf{u}},\phi)=\left\{\left.\mbf{h}\in \ds{R}^M \right|\sqrt{r}\leq \|\mbf{h}\| < \sqrt{R}, ~ \angle(\mbf{h},\tilde{\mbf{u}})<\phi \right\},\]
where in the terminology of Section \ref{S_Panj}, $\tilde{\mbf{u}}$ is the \emph{quantized direction} and $r$ is the \emph{quantized magnitude}, which is denoted as $\tilde{Y}$ in \eqref{quant-mag}.

\begin{figure}
\centering
\includegraphics[width=1.1in]{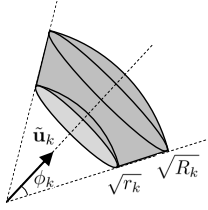}
\caption{Sector-type channel uncertainty region $S_k=S(R_k,r_k,\tilde{\mathbf{u}}_k,\phi_k)$ for user $k$.}
\label{Fig1}
\end{figure}

For a specific point in time, consider the channel realizations $\mbf{H}=[\bh_1^T,\bh_2^T,\cdots,\bh_M^T]$ and the corresponding quantization (or channel uncertainty) regions \[S_k=S(R_k,r_k,\tilde{\mathbf{u}}_k,\phi_k)\defined \mc{S}_k(\bh_k),\] where the \emph{quantization functions} $\mc{S}_k(\bh_k)$, $1{\leq} k{\leq} M$, are defined in Section \ref{S_Se}. Also, let $I_k$ denote the corresponding activity flags.

The goal is to optimize the power control function for the robust design problem \eqref{robust-design}. Therefore, for the current channel realizations $\mbf{H}$, we have to find the transmission power levels $P_k$ that minimize the instantaneous sum power subject to the worst-case SINR constraints:
\begin{align}\label{robust-prob}
&\min\limits_{P_k} \sum_{k=1}^{M}{P_k}\\
&~\textmd{s.t.} \inf_{\mathbf{w}\in
S_k}\frac{P_k\left|\mathbf{w}^{T}\mathbf{v}_k\right|^2}{\sum\limits_{l\neq
k}{P_l\left|\mathbf{w}^{T}\mathbf{v}_l\right|^2}+1}\geq
\gamma_k I_k,~~k{=}1{,}2{,}{\cdots}{,}M\label{robust-const}
\end{align}
where the beamforming vectors $\mbf{v}_k$ are fixed, since the quantized directions $\tilde{\mathbf{u}}_k$ are fixed.
Let us refer to the users with $I_k=0$ as the \emph{silent} users and the users with $I_k=1$ as the \emph{active} users and let the set $\mc{K}$ denote the set of active users: $\mc{K}=\{1{\leq} k{\leq} M| I_k=1\}$. In general $\mc{K}$ is a random set depending on the channel realizations and the specified outage regions. Now, considering the power control problem in \eqref{robust-prob}, we note that if a user $k$ is silent, i.e. $I_k=0$, the corresponding SINR constraint in \eqref{robust-const} is redundant as the problem returns $P_k=0$ for such a user. We therefore confine the SINR constraints in \eqref{robust-const} to the set of active users, i.e. the indices $k\in\mc{K}$.

According to the robust design formulation in \eqref{robust-design}, for any channel realization $\mbf{H}=[\bh_1^T,\bh_2^T,\cdots,\bh_M^T]$, the base-station is required to serve (guarantee the target SINR's for) all active users $k\in\mc{K}$. The power control problem in \eqref{robust-prob} therefore must be feasible for the active users. The following theorem presents a sufficient condition that guarantees feasibility.

\begin{theorem}\label{T_Haft}
To ensure the feasibility of the robust power control problem in \eqref{robust-prob}, it is sufficient to have the following for all $1\leq k\leq M$:
\be\label{MQCS}
\NkD\geq \left({4\lambda_M}\Big{/}{\sin\left(\arctan\left(\frac{\sin{\theta_k^\circ}}{1+\sqrt{(M-1)\gamma_k}} \right)\right)}\right)^{M-1}.
\ee
\end{theorem}
\begin{proof}
See Appendix \ref{A_1}.
\end{proof}
This condition is referred to as the minimum quantization codebook size (MQCS) condition in the remainder of this paper.

In the high resolution regime, where the codebook sizes tend to infinity, the MQCS conditions are clearly satisfied; therefore, feasibility of the power control problem is not an issue as far as the high resolution analysis is considered. This condition however plays a key role in finding the minimum number of feedback bits $B$ for which the asymptotic bit allocation laws are applicable. This issue is discussed in further detail in Theorem \ref{T_Dah} of Section \ref{S_Shish}.

We are now ready to solve the power control problem in \eqref{robust-prob}. An exact numerical solution to this problem can be obtained by transforming it into a semidefinite programming (SDP) problem as described in the following.

\begin{theorem}\label{T_Shish}
The problem in \eqref{robust-prob} is equivalent to the following SDP problem for $M\geq3$:
\begin{align}\label{SDP}
&\!\min\limits_{P_k,\lambda_k,\mu_k} \sum_{k\in\mc{K}}{P_k}\\
&~~\textmd{s.t.} \quad \frac{1}{\gamma_k}P_k\mbf{v}_k\mbf{v}_k^T{-}\!\sum_{l\neq k}{P_l\mbf{v}_l\mbf{v}_l^T}{\succeq} (\lambda_k{-}\mu_k)\mbf{I}_M{+}\frac{\mu_k}{\cos^2{\phi_k}}\tilde{\mathbf{u}}_k\tilde{\mathbf{u}}_k^T\nonumber\\
&\qquad ~~~ \lambda_k\geq \frac{1}{r_k}, ~ \mu_k\geq 0, ~ P_k\geq 0,~ k\in\mc{K},\nonumber
\end{align}
where $\bf{I}_M$ is the $M{\times} M$ identity matrix.
\end{theorem}
\begin{proof}
The proof is based on the Polyak's theorem in \cite{polyak}. See Appendix \ref{A_2} for details.
\end{proof}

Although the SDP reformulation provides a numerically efficient solution, it does not give the minimized sum power in a closed form. The availability of a closed-form expression for sum power is crucial in optimizing the quantization codebook structures. Moreover, for practical systems where the base-station needs to continuously compute and update the transmission powers for each fading block, the SDP reformulation would be of a limited application. We therefore resort to a suboptimal solution to the problem in \eqref{robust-prob} that provides a closed-form upper bound to the sum power. This upper bound solution is used later in Section \ref{S_Shish} as the objective function for optimization of the users' quantization codebooks.

First, we bound the SINR terms as follows. For the sector-type regions $S_k=S(R_k,r_k,\tilde{\mathbf{u}}_k,\phi_k)$, we have
\begin{align}
\!\!\!\!\inf_{\mbf{w}\in S_k}\frac{P_k\left|\mbf{w}^{T}\mathbf{v}_k\right|^2}{\sum\limits_{l\neq k}{P_l\left|\mbf{w}^{T}\mathbf{v}_l\right|^2}+1}
&\stackrel{(\md{a})}{=}
\inf_{\mbf{w}\in S_k}\frac{P_k r_k|\mathbf{\hat{w}}^{T}\mathbf{v}_k|^2}{r_k\!\!\sum\limits_{l\neq k}{P_l|\mathbf{\hat{w}}^{T}\mathbf{v}_l|^2}{+}1}\nonumber\\
&{\geq} \frac{P_k r_k\inf_{\mbf{w}\in S_k}|\mathbf{\hat{w}}^{T}\mathbf{v}_k|^2}{r_k\!\!\sum\limits_{l{\neq} k}{P_l\sup_{\mbf{w}\in S_k}|\mathbf{\hat{w}}^{T}\mathbf{v}_l|^2}{+}1}\label{comp3}\\
&{\stackrel{(\md{b})}{=}} \frac{P_kr_k\sin^2{(\theta_k{-}\phi_k)}}{(\sum_{l{\neq}k}{P_l})r_k\sin^2{\phi_k}{+}1},\label{comp2}
\end{align}
where $\hat{\mbf{w}}=\mbf{w}/\|\mbf{w}\|$. The equality (a) holds since the SINR term is monotonic in $\|\mbf{w}\|$, i.e. the minimum occurs on the spherical boundary region $\|\bw\|{=}\sqrt{r_k}$ in Fig. \ref{Fig1}. The equality (b) holds since $\mbf{v}_k$'s are the zero-forcing directions for $\tilde{\mathbf{u}}_k$'s. To see this, consider $\theta_k=\angle(\tilde{\mathbf{u}}_k,\tilde{\mathbf{U}}_{-k})$ as defined earlier. By considering the zero-forcing principle, we have $\angle(\bv_k,\tilde{\bu}_k)=\frac{\pi}{2}-\theta_k$ and $\angle(\bv_l,\tilde{\bu}_k)=0$ for $l\neq k$. Now, noting the definition of the sector-type region $S_k$, we have \begin{align}
&\max_{\bw \in\mc{S}_k} {\angle(\bw,\bv_k)}= \frac{\pi}{2}-\theta_k +\phi_k,\nonumber\\
&\min_{\bw \in\mc{S}_k} {\angle(\bw,\bv_l)}= \frac{\pi}{2}-\phi_k.\nonumber
\end{align}
By substituting the cosine of these angles in the numerator and denominator of \eqref{comp3}, we achieve the final expression in \eqref{comp2}.

In order to obtain an upper bound on the sum power, we set the last term in \eqref{comp2} to be equal to $\gamma_k$: \[\frac{P_k r_k\sin^2{(\theta_k{-}\phi_k)}}{(\sum_{l{\neq}k}{P_l})r_k\sin^2{\phi_k}{+}1}=\gamma_k.\] This is a set of linear equations in $P_k$, $k\in \mc{K}$, where $\mc{K}$ is the set of active users. By solving these equations and computing $\sum_{k\in\mc{K}}{P_k}$, we achieve the following upper bound for the sum power:
\be\label{Pub}
P_{_{\md{MU}}}\defined\sum_{k\in\mc{K}}{P_k}=\frac{\sum_{k\in\mc{K}}{\alpha_k/\beta_k}}{1-\sum_{k\in\mc{K}}{\alpha_k}},
\ee
where $\alpha_k=\left(1+\frac{\sin^2{(\theta_k-\phi_k)}}{\gamma_k\sin^2{\phi_k}}\right)^{-1}$ and $\beta_k=r_k\sin^2{\phi_k}$. The subscript $\md{MU}$ in $P_{_\md{MU}}$ stands for \emph{multi-user}. The closeness of the upper bound solution in \eqref{Pub} and the solution to the SDP problem in \eqref{SDP} is verified numerically in Section \ref{S_Hasht}.

The upper bound in \eqref{Pub} is a bound on the instantaneous sum power for a single snapshot of channel realizations in time. As the users' channels change over time, the quantized magnitudes $r_k=\tilde{Y}_k$, the quantized directions $\tilde{\bu}_k$ and the corresponding angles $\theta_k=\angle(\tilde{\mathbf{u}}_k,\tilde{\mathbf{U}}_{-k})$ all change with time. The variables $\alpha_k$ and $\beta_k$ in \eqref{Pub} are therefore random variables. Since we are interested in the expected value of the sum power as the design objective in \eqref{robust-design}, we use the following sum-power upper bound approximation so that the expectation operation can be applied conveniently.
\begin{theorem}\label{T_Hasht}
In the asymptotic regime with large quantization codebook sizes and small values of $\phi_k$, we have
\be\label{PubApprx}
P_{_{\md{MU}}}=\sum_{k\in\mc{K}}{e_k}+\sum_{k\in\mc{K}}{f_k\phi_k}+\sum_{k\in\mc{K}}{o(\phi_k)},
\ee
where
\be\label{e-f-zeta}
e_k=\frac{\gamma_k}{r_k}(1+\zeta_k^2),~~~f_k=\frac{2\gamma_k}{r_k}(\zeta_k+\zeta_k^3),~~~\zeta_k=\cot\theta_k.~~
\ee
Here, the notation $g(\phi)=o(\phi)$, for an arbitrary function $g(\cdot)$, means that ${\lim_{\phi\rightarrow0}{g(\phi)/\phi}}=0$.
\end{theorem}
\begin{proof}
See Appendix \ref{A_3}.
\end{proof}

So far, we have only considered the active users $k\in\mc{K}$. In order to make the results applicable to the general case where some users might be in outage, we substitute $\gamma_k$ with $\gamma_k I_k$ in definitions of $e_k$ and $f_k$ in \eqref{PubApprx} so that users with $I_k=0$ contribute zero power to the sum-power upper bound. We therefore use the following expression for the average sum-power upper bound, where we have also replaced $r_k$ by the quantized magnitude $\tilde{Y}_k$:
\be\label{PubApprx2}
\mbb{E}[P_{_{\md{MU}}}]\approx\sum_{k}\mbb{E}[{e_k}]+\sum_{k}{\mbb{E}[f_k]\phi_k},
\ee
where $\zeta_k{=}\cot\theta_k$ and $e_k$ and $f_k$ are redefined as follows:
\be\label{e-f-zeta2}
e_k=\frac{\gamma_k I_k}{\tilde{Y}_k}(1+\zeta_k^2),~~~f_k=\frac{2\gamma_k I_k}{\tilde{Y}_k}(\zeta_k+\zeta_k^3).\nonumber
\ee

This concludes the optimization of the power control function. In the next section, we use the average sum-power upper bound in \eqref{PubApprx2} to optimize the product quantization codebook structures and to derive the asymptotic bit allocation laws.

\section{Product Codebook Optimization and Asymptotic Bit Allocation Laws}\label{S_Shish}

In this section, we study the quantization codebook optimization. For this purpose, we use the average transmission power bound in \eqref{PubApprx2} in order to optimize the users' magnitude and direction codebook sizes for a given feedback link capacity constraint and to derive the optimal bit allocation across the users and their magnitude and direction quantization codebooks. The optimization process is asymptotic in the feedback rate $B$ and assumes large quantization codebook sizes, $\NkM,\NkD\gg 1$.

Consider the sum-power upper bound in \eqref{PubApprx2}. Assuming that $\theta_k=\angle(\tilde{\bu}_k,\tilde{\mbf{U}}_{-k})$ is approximately uniform in $[0,\pi/2]$ and using the definition of the activity flag in \eqref{flag2}, we have
\begin{align}
&\mbb{E}[e_k]= \gamma_k \mbb{E}\left[{1}/{\tilde{Y}_k}\right]\frac{1}{\pi/2}{\int_{\theta_k^\circ}^{\pi/2} \!\!\!\!{1{+}\zeta_k^2}\ \du \theta_k} \approx \frac{2\gamma_k}{\pi\theta_k^\circ}\mbb{E}\left[{1}/{\tilde{Y}_k}\right], \nonumber \\
&\mbb{E}[f_k] = 2\gamma_k \mbb{E}\left[{1}/{\tilde{Y}_k}\right]\frac{1}{\pi/2}{\int_{\theta_k^\circ}^{\pi/2}
\!\!\!\!{\zeta_k{+}\zeta_k^3} \ \du \theta_k} \approx \frac{2\gamma_k}{\pi{\theta_k^\circ}^2}\mbb{E}\left[{1}/{\tilde{Y}_k}\right],\nonumber
\end{align}
where the approximations hold for $\theta_k^\circ{\ll} 1$. By substituting these in \eqref{PubApprx2}, we achieve
\be\label{PubApprx4} \mbb{E}[P_{_{\md{MU}}}]\approx  \frac{2}{\pi} \sum_{k=1}^{M} {{\frac{\gamma_k}{\theta_k^\circ}\ \mbb{E}\left[\frac{1}{\tilde{Y}_k}\right]
\left(1+\frac{\phi_k}{\theta_k^\circ}\right)}}. \ee

The parameter $\phi_k$ in \eqref{PubApprx4} is controlled by the direction codebook size $\NkD$ as described by \eqref{phi-bound-mu}. The term $\mbb{E}[{1}/{\tilde{Y}_k}]$ on the other hand is controlled by the magnitude codebook. The asymptotically optimal codebook that minimizes this term is the uniform (in dB) magnitude quantization codebook $\mbb{Y}^\star$. By setting $\mbb{Y}=\mbb{Y}^\star$ and by using \eqref{Pavescale-mu} and \eqref{phi-bound-mu}, we can bound the average sum power in \eqref{PubApprx4} as follows:
\begin{align}\label{final-power-bound}
&\mathbb{E}\left[P_{_{\md{MU}}}\right]<\frac{2\rho_{_{\md{MU,CSI}}}}{\pi}\ \cdot\nonumber\\
&~~~\sum_{k{=}1}^M{\frac{\gamma_k}{\theta_k^\circ}
\left(1{+}\NkM^{{-}\zeta_k(\NkM)}{+}\omega\NkM^{{-}2\zeta_k(\NkM)}\right)
\left(1{+}\frac{4\lambda_M}{\theta_k^\circ}\NkD^{{-}\frac{1}{M{-}1}}\right)}
\end{align}
where $\rho_{_{\md{MU,CSI}}}$ is defined in \eqref{rho-CSI} and the variable $\omega$ and the function $\zeta_k(\cdot)$ depend on the magnitude outage probability and the distribution of the channel magnitude as described in Section \ref{S_Panj-A-1}.

Our goal is to minimize the average sum-power upper bound in \eqref{final-power-bound} in terms of the magnitude and direction quantization codebook parameters. The optimization constraints are as follows. Assuming a total number of feedback bits $B$, we have the following constraint on the codebook sizes:
\be\label{rate-constraint} \prod_{k=1}^{M}{N_k}=\prod_{k=1}^{M}{\left(\NkM\NkD+1\right)}= 2^B.\ee
The two other constraints are the target outage probability constraints given by \eqref{total-outage} and the MQCS conditions in \eqref{MQCS}. For the total feedback rate $B$, the target outage probabilities $q_k$, and the target SINR values $\gamma_k$, the product codebook optimization problem is therefore formulated as follows:
\begin{align}\label{codebook-opt-mu}
&\min_{\substack{\NkM,\NkD\\ y_{k}^{(1)}, \theta_k^\circ}} \sum_{k{=}1}^M{\frac{\gamma_k}{\theta_k^\circ}
\!\left(\!1{+}\NkM^{{-}\zeta_k(\NkM)}{+}\omega\NkM^{{-}2\zeta_k(\NkM)}\!\right)\!\!\!
\left(\!1{+}\frac{4\lambda_M}{\theta_k^\circ}\NkD^{{-}\frac{1}{M{-}1}}\!\right)} \\
&~~~\md{s.t.} ~\prod_{k=1}^{M}{\left(\NkM\NkD+1\right)}= 2^B, \label{const1}\\
& ~~~~~~~~ F(y_{k}^{(1)})+\frac{2}{\pi}\theta_k^\circ \leq q_k,\label{const2}\\
& ~~~~~~~~ \NkD {\geq} \left({4\lambda_M}\Big{/}{\sin\left(\arctan\left(\frac{\sin{\theta_k^\circ}}{1{+}\sqrt{(M{-}1)\gamma_k}} \right)\right)}\right)^{M{-}1}\!\!.\label{const3}
\end{align}
In order to obtain a closed-form solution for the optimal product structure, we simplify this problem as follows.

First, by assuming $\NkM,\NkD\gg 1$ and using the fact that $\lim_{\NkM\rightarrow\infty}{\zeta_k(\NkM)}=1$, we use the following approximation for the objective function:
\begin{align}\label{obj-approx}
&\left(1+\NkM^{-\zeta_k(\NkM)}+\omega\NkM^{-2\zeta_k(\NkM)}\right)
\left(1+\frac{4\lambda_M}{\theta_k^\circ}\NkD^{-\frac{1}{M-1}}\right)\nonumber\\
&\hspace{1in}\approx \left(1+\NkM^{-1}\right)
\left(1+\frac{4\lambda_M}{\theta_k^\circ}\NkD^{-\frac{1}{M-1}}\right)\nonumber\\
&\hspace{1in}\approx 1+\NkM^{-1}+\frac{4\lambda_M}{\theta_k^\circ}\NkD^{-\frac{1}{M-1}}.
\end{align}
We therefore have the following approximate upper bound:
\be\label{approx-final-power-bound}
\mathbb{E}\left[P_{_{\md{MU}}}\right]<\frac{2\rho_{_{\md{MU,CSI}}}}{\pi}\sum_{k{=}1}^M{\frac{\gamma_k}{\theta_k^\circ}
\left(1{+}\NkM^{-1}{+}\frac{4\lambda_M}{\theta_k^\circ}\NkD^{{-}\frac{1}{M{-}1}}\right)}.
\ee

Next, we simplify the optimization constraints as follows. We approximate the first constraint \eqref{const1} as $\prod_{k=1}^M{\NkM\NkD}=2^B$. Regarding the outage constraint in \eqref{const2}, one can easily show that the objective function in \eqref{codebook-opt-mu} is a decreasing function of $y_{k}^{(1)}$ and $\theta_k^\circ$. The constraint in \eqref{const2} should therefore be satisfied with equality at the optimum. In order to simplify this constraint, we make the assumption that the magnitude and direction outage probabilities are equally likely\footnote{It can be shown that any other division of the form $\dot{q}_k = \alpha q_k$ and $\ddot{q}_k = (1-\alpha) q_k$ with $0<\alpha<1$ only changes the bit allocations in Theorem \ref{T_Noh} by a finite constant and therefore does not affect the asymptotic bit allocation results.} :
\begin{align}
&\dot{q}_k=F(y_{k}^{(1)})=\frac{q_k}{2}\label{qdot}\\
&\ddot{q}_k=\frac{2}{\pi}\theta_k^\circ=\frac{q_k}{2}\label{qddot}.
\end{align}
According to this assumption, $y_{k}^{(1)}=F^{-1}(q_k/2)$ and \be\label{theta0} \theta_k^\circ=\frac{\pi}{4}q_k\ee are fixed and the codebook optimization is only over the codebook sizes. Finally, since the optimization is asymptotic in the codebook sizes $\NkM$ and $\NkD$, the last constraint in \eqref{const3} is redundant. This constraint however is used later to derive a lower bound on the total feedback rate $B$ such that the target outage probabilities are feasible.

Now, by using the approximation in \eqref{obj-approx}, the optimization problem in \eqref{codebook-opt-mu} simplifies to the following optimization problem:
\begin{align}\label{approx-codebook-opt-mu}
&\min_{\NkM,\NkD} \sum_{k=1}^M{\frac{\gamma_k}{\theta_k^\circ}
\left(1+\NkM^{-1}+\frac{4\lambda_M}{\theta_k^\circ}\NkD^{-\frac{1}{M-1}}\right)} \\
&~~\md{s.t.} ~~ \prod_{k=1}^{M}{\NkM\NkD}= 2^B.\label{approx-const1}
\end{align}

Define $\BkM{\defined}\log{\NkM}$ and $\BkD{\defined}\log{\NkD}$ as the number of magnitude and direction quantization bits respectively.
\begin{theorem}\label{T_Noh}
Define
\begin{align}
&\BbarM=\frac{1}{M^2}B-\frac{M{-}1}{M}\log{\frac{1}{\qbar}}-\kappa_{_{\md{MU}}} \label{B-barM}\\
&\BbarD=\frac{M{-}1}{M^2}B+\frac{M{-}1}{M}\log{\frac{1}{\qbar}}+\kappa_{_{\md{MU}}}, \label{B-barD}
\end{align}
where $\kappa_{_{\md{MU}}}{=}\frac{M-1}{M}\log{\frac{16\lambda_M}{\pi(M-1)}}$ and $\qbar{=}\left({\prod_{k}{q_k}}\right)^{1/M}$ is the geometric mean of the target outage probabilities.
The optimal values of $\BkM$ and $\BkD$ are given by
\begin{align}
\BkM &=\BbarM+\log{\frac{\gamma_k}{\gbar}}+\log{\frac{\qbar}{q_k}}\label{OptimalBits1}\\
\BkD &=\BbarD+{{(M{-}1)}\log{\frac{\gamma_k}{\gbar}}}+{{2(M{-}1)}\log{\frac{\qbar}{q_k}}},\label{OptimalBits2}
\end{align}
where  $\gbar{=}\left({\prod_{k}{\gamma_k}}\right)^{1/M}$ is the geometric mean of the target SINR values.
\end{theorem}
\begin{proof}
See Appendix \ref{A_4}.
\end{proof}

Although the bit allocation laws in Theorem \ref{T_Noh} are derived with the simplifying assumption of equal magnitude and direction outage probabilities in \eqref{qdot} and \eqref{qddot}, the numerical results, as shown in Fig. \ref{Fig2}, verify that the analytical results in \eqref{OptimalBits1} and \eqref{OptimalBits2} are close to the bit allocations derived by numerical minimization of \eqref{codebook-opt-mu}. For the example in this figure, the base-station has $M=3$ antennas and serves three users with the target parameters $\gamma_k=15$dB and $q_k=0.02$ for the first user ($k=1$) and $\gamma_k=10$dB and $q_k=0.05$ for the two other users ($k=2,3$). The user channels are i.i.d. and $\bh_k\sim\mc{N}(0,\mbf{I}_M)$, where $\bf{I}_M$ is the $M{\times} M$ identity matrix. Also, the number of bits are rounded to the closest integer numbers\footnote{This is a popular approach in solving integer programming problems, where the integer conditions are relaxed and the optimum of the relaxed problem is rounded to the closed integer\cite{boyd}.}.

\begin{figure}%
\vspace{-0.7cm}
\centering
\includegraphics[width=3.4in]{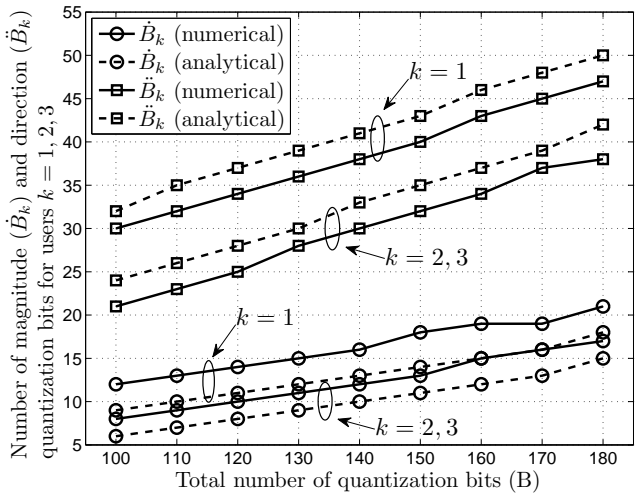}
\caption[]{Analytical magnitude and direction bit allocations in \eqref{OptimalBits1} and \eqref{OptimalBits2} vs. numerical bit allocations for three users with $\gamma_1=15$dB, $q_1=0.02$ and $\gamma_2=\gamma_3=10$dB and $q_2=q_3=0.05$. }
\label{Fig2}
\vspace{-0.7cm}
\end{figure}

\begin{corollary}
For each user $k$, the optimal number of magnitude and direction quantization bits are related as follows:
\be\label{mag-dir-relation} \BkD=(M{-}1)\BkM+(M{-}1)\log{\frac{1}{q_k}}+M\kappa_{_{\md{MU}}},\ee
where $\kappa_{_{\md{MU}}}$ is defined in Theorem \ref{T_Noh}. Moreover the total number of quantization bits for user $k$ is given by
\be\label{eachuser} B_k=\BkM+\BkD=\frac{1}{M}B+M\log{\frac{\gamma_k}{\gbar}}+{{(2M{-}1)}\log{\frac{\qbar}{q_k}}}.\ee
\end{corollary}

As it is expected, if the users are homogenous in their requested target parameters, i.e. $q_k$ and $\gamma_k$ are the same for all users, each user takes an equal share of $\frac{1}{M} B$ of the total feedback rate. In the case of heterogenous users, on the other hand, a user with a higher QoS (lower target outage probability) and a higher target downlink rate (higher target SINR) uses a higher feedback rate $B_k$.


The bit allocation laws in Theorem \ref{T_Noh} are asymptotic results in the feedback rate $B\rightarrow\infty$. In the following, in order to get a sense of how high the feedback rate should be, we determine a lower bound on $B$ for which the target SINR values $\gamma_k$ are in fact achievable with the target outage probabilities $q_k$.
\begin{theorem}\label{T_Dah}
Assume $\gamma_k>1$ and $q_k\ll1$ and define \be\label{Qind} \mc{Q}_k\defined \frac{{\sqrt{\gamma_k}}}{q_k}.\nonumber\ee For the target SINR's $\gamma_k$ to be satisfied with outage probabilities $q_k$, the following total feedback rate $B$ is sufficient:
\be\label{total-rate} B>\frac{1}{2}M^2\log{\gbar} + (M^2{-}M)\log{\frac{1}{\qbar}}+M^2\log\Delta+b,\ee
Here $\gbar$ and $\qbar$ are the geometric means of $\gamma_k$'s and $q_k$'s respectively and \[\Delta=\frac{\bar{\mc{Q}}}{\min\limits_{1{\leq} k{\leq} M}{\mc{Q}_k}},\] where $\bar{\mc{Q}}$ is the geometric mean of $\mc{Q}_k$'s. The constant $b$ in \eqref{total-rate} is defined as $b{=}\frac{1}{2}M^2{+}\frac{3}{2}M^2\log M{+}M^2\kappa_{_{\md{MU}}}$, where $\kappa_{_{\md{MU}}}$ is defined in Theorem \ref{T_Noh}.
\end{theorem}
\begin{proof}
See Appendix \ref{A_5}.
\end{proof}

Several interesting results can be extracted from Theorem \ref{T_Dah}. First, we observe that for the QoS constraints to remain feasible, the system feedback link capacity should scale logarithmically with the geometric mean of the target SINR values and the geometric mean of the inverse target outage probabilities. Second, if we compare the case of homogenous users with the case of heterogenous users, we see that heterogenous users impose an additional requirement, $M^2\log\Delta$, on the total feedback rate. If we think of $\mc{Q}_k$'s as users' QoS indicators, the variable $\Delta$ can be interpreted as a measure of discrepancy among users' QoS requirements. A higher QoS discrepancy requires a higher feedback bandwidth.

Finally, in order to study the performance of the limited-feedback system as the feedback rate increases, we substitute the optimal magnitude and direction codebook sizes given by Theorem \ref{T_Noh} into the average sum-power upper bound in \eqref{approx-final-power-bound}. The following theorem shows the scaling of the average sum power with the feedback rate $B$.
\begin{theorem}\label{T_Yazdah} For a limited-feedback system with a total number of $B$ feedback bits, we have
\be\label{power-scaling}
\mathbb{E}\left[P_{_{\md{MU}}}\right] <P_{_{\md{MU,CSI}}}\left(1+{\frac{\sigma_{_{\md{MU}}}}{\qbar}}\cdot{2^{-\frac{B}{M^2}}}\right),
\ee
where $P_{_{\md{MU,CSI}}}$ is defined in \eqref{P-CSI}, $\qbar$ is the geometric mean of the target outage probabilities, and \[\sigma_{_{\md{MU}}}=\frac{16M}{\pi(M{-}1)}\left(\frac{\pi^{3/2}(M-1)\Gamma((M+1)/2)}{{16}{\Gamma(M/2)}}\right)^{{1}/{M}}.\]
\end{theorem}
\begin{proof}
See Appendix \ref{A_6}.
\end{proof}

If we define a quantization distortion measure as
\be\label{distortion}
D(B)=\frac{\mathbb{E}\left[P_{_{\md{MU}}}\right]-P_{_{\md{MU,CSI}}}}{P_{_{\md{MU,CSI}}}},
\ee
Theorem \ref{T_Yazdah} implies that the distortion measure scales as $2^{-\frac{B}{M^2}}$ as $B\rightarrow\infty$.

\section{Numerical Validation}\label{S_Hasht}

This section presents the numerical results that support and verify the analytical results in the earlier sections.

\begin{figure}
\centering
\includegraphics[width=3.4in]{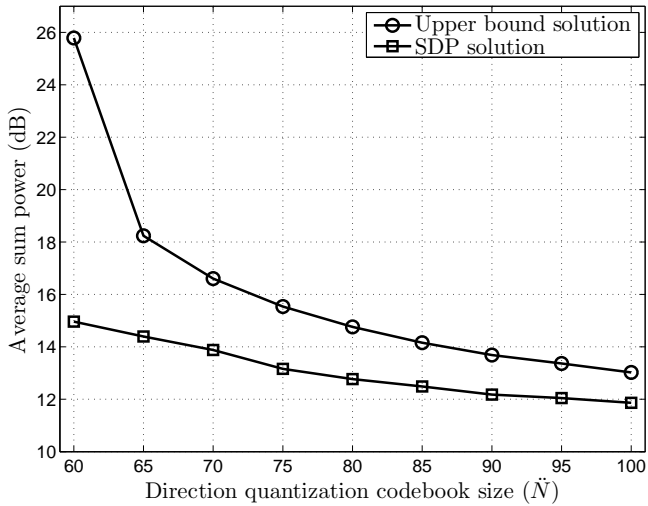}
\caption{SDP solution in \eqref{SDP} vs. the upper bound solution in \eqref{Pub} for three users with target SINR's of $\gamma_1=3$dB and $\gamma_2,\gamma_3=6$dB.}
\label{SDP_ub_comp}
\end{figure}

\subsubsection{Upper Bound Approximation for The SDP Problem in \eqref{SDP}}

The codebook optimizations in Section \ref{S_Shish} are based on the sum-power upper bound solution in \eqref{Pub} as an approximation of the solution to the SDP problem in \eqref{SDP}. Here we investigate the accuracy of this approximation by comparing the two solutions for $M=3$ users with target SINR's of $\gamma_1=3$dB and $\gamma_2,\gamma_3=6$dB.

To simplify the comparison, we assume perfect channel magnitude information, i.e. the quantized magnitude variables $r_k$ in \eqref{SDP} and \eqref{Pub} are equal to the exact channel magnitudes. For direction quantization we use Grassmannian codebooks from \cite{grass_code}. The same codebook size is used for all users. The sum power values are averaged over $100$ channel realizations for which the SDP problem is feasible. The user channels are i.i.d. and $\bh_k\sim\mc{N}(0,\mbf{I}_M)$.

Fig. \ref{SDP_ub_comp} compares the two solutions as a function of the direction quantization codebook size $\ND$. As the figure shows, the two solutions converge as $\ND$ increases\footnote{This can made rigorous by showing that as $\NkD\rightarrow\infty$ and $\phi_k\rightarrow 0$, the inequality in \eqref{comp3} is satisfied with an equality and therefore the upper bound solution is exact in the asymptotic high-resolution regime.}. This justifies the use of the upper bound solution as the optimization objective in our high resolution analysis.

\subsubsection{Bit Allocations and Distortion Scaling}

Fig. \ref{Fig2} in Section \ref{S_Shish} compares the numerical and analytical bit allocations for $M=3$ users with different target parameters. Here we repeat the process with a different set of parameters and record the share of each user $B_k$ from the total number of feedback bits $B$. The target parameters are $\gamma_1{=}2$dB, $\gamma_2{=}5$dB, $\gamma_3{=}8$dB, and $q_1{=}q_2{=}q_3{=}0.1$. Channel models are similar to those used in Figs. \ref{Fig2} and \ref{SDP_ub_comp} and the bit allocations are rounded to the closest integer numbers. As the figure verifies, users with higher target SINR's receive larger shares of the total feedback rate.

\begin{figure}
\centering
\includegraphics[width=3.4in]{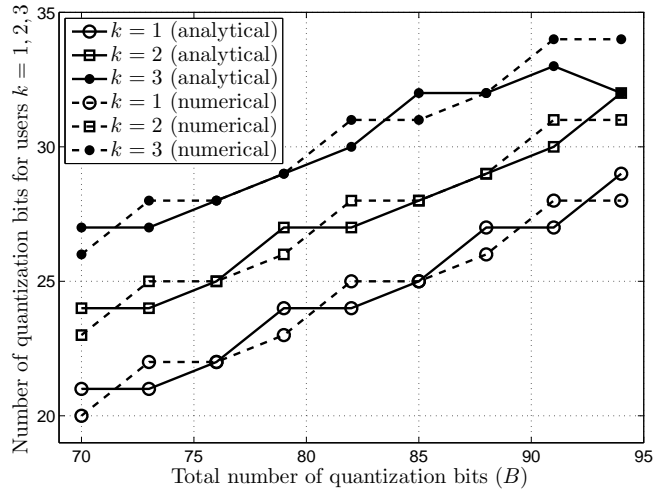}
\caption{Bit allocations $B_k$ for three users with $\gamma_1{=}2$dB, $\gamma_2{=}5$dB, $\gamma_3{=}8$dB, and $q_1{=}q_2{=}q_3{=}0.1$.}
\label{bit_allocations}
\end{figure}

Finally, we investigate the system performance scaling with the number of feedback bits for the same set of parameters as in Fig. \ref{bit_allocations}. For this purpose we use the average sum-power upper bound $\mathbb{E}\left[P_{_{\md{MU}}}\right]$ in \eqref{final-power-bound} and the definition of the distortion measure $D(B)$ in \eqref{distortion}. Fig. \ref{distortion_fig} shows the distortion measure as a function of $B$ when numerical and analytical bit allocations are utilized. As expected, the two bit allocations show close performances. The figure also shows the distortion upper bound in Theorem \ref{T_Yazdah} for the purpose of comparison.

\begin{figure}
\centering
\includegraphics[width=3.4in]{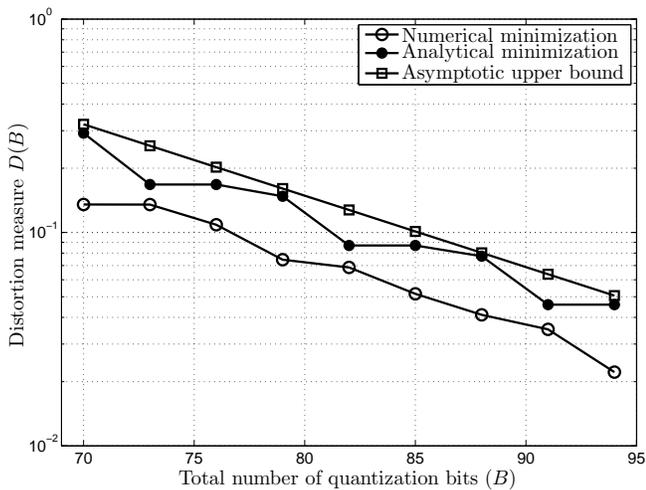}
\caption{Distortion measure with numerical and analytical bit allocations for three users with $\gamma_1{=}2$dB, $\gamma_2{=}5$dB, $\gamma_3{=}8$dB, and $q_1{=}q_2{=}q_3{=}0.1$.}
\label{distortion_fig}
\end{figure}

\section{Concluding Remarks}\label{S_Haft}

We conclude this paper by comparing the asymptotic magnitude-direction bit allocation law for the multi-user system with that of the single-user system discussed in \cite{PartI}. For the multi-user system and in the asymptotic regime where $B\rightarrow\infty$, the relation in \eqref{mag-dir-relation} implies that the number of magnitude and direction quantization bits (for each user) are related as follows: \be\label{mu-opt-bits} \ddot{B}_{_{\md{MU}}}= (M-1)\dot{B}_{_{\md{MU}}}.\ee The subscript $\md{MU}$ in \eqref{mu-opt-bits} stands for multi-user. For single-user systems, on the other hand, we have the following bit allocation law \cite{PartI}: \be\label{su-opt-bits} \ddot{B}_{_{\md{SU}}}= \frac{M-1}{2}\dot{B}_{_{\md{SU}}},\ee with $\md{SU}$ standing for single-user. If we define a relative quantization resolution as $\Upsilon=\ddot{B}/\dot{B}$, then we have \be\label{rel-res} \Upsilon_{_{\md{MU}}}=2\ \Upsilon_{_{\md{SU}}},\ee
which means that for the same number of magnitude quantization bits, the number of multi-user direction quantization bits is twice the number of single-user direction quantization bits. This is shown schematically in Fig. \ref{su_mu_bit_allocation}. As the figure implies, in order to make a single-user channel quantization codebook applicable to the multi-user system, each direction quantization region in the single-user codebook should be further quantized with the same resolution as the whole unit hypersphere.


As the final note we mention that the results in this paper are based on real channel space assumption mainly for the ease of geometric representation of the quantization regions and the corresponding sum-power calculations in Section \ref{S_Char}. However, the exact same approach introduced in this paper can be applied to complex space channels. The only main difference is in using the upper bound in \eqref{phi-bound-mu} for real Grassmannian codebooks. For complex Grassmannian codebooks we need to use the following bound \cite{heath03}:
\be\label{complex-bound} \sin\phi_k<2\NkD^{-\frac{1}{2(M-1)}},\ee
By doing so, the rest of the analysis can be applied in a similar fashion to derive the bit allocation laws. In particular, the codebook optimization problem in \eqref{approx-codebook-opt-mu} translates to the following problem for complex channels:
\begin{align}\label{approx-codebook-opt-mu-complex}
&\min_{\NkM,\NkD} \sum_{k=1}^M{\frac{\gamma_k}{\theta_k^\circ}
\left(1+\NkM^{-1}+\frac{2}{\theta_k^\circ}\NkD^{-\frac{1}{2(M-1)}}\right)} \nonumber\\
&~~\md{s.t.} ~~\prod_{k=1}^{M}{\NkM\NkD}= 2^B.\nonumber
\end{align}
By solving this problem, one can easily derive the bit allocation laws and the system performance scaling for the complex space similar to the ones in Theorems \ref{T_Noh} and \ref{T_Yazdah}.

In particular, due to the difference between \eqref{phi-bound-mu} and \eqref{complex-bound} in the exponents of the direction codebook sizes $\NkD$, the asymptotic magnitude-direction bit allocation in the complex space turns out to be $\BD= 2(M-1)\BM$ for the multi-user case and $\BD= (M-1)\BM$ for the single-user case. Therefore, the complex-space quantization resolutions also satisfy the relation in \eqref{rel-res}.

\begin{figure}[!t]
\centering
\includegraphics[width=2in]{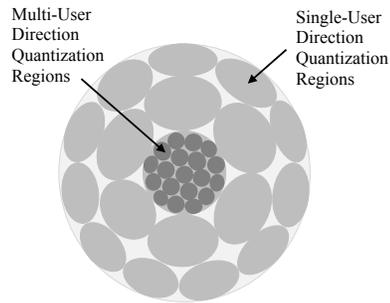}
\caption{$\Upsilon_{_{\md{MU}}}{=}2 \Upsilon_{_{\md{SU}}}$; single-user and multi-user direction quantization regions are shown as spherical caps on the unit hypersphere.}
\label{su_mu_bit_allocation}
\end{figure}

%
%
%
%
%

\appendix


\subsection{Proof of Theorem \ref{T_Haft}} \label{A_1}

In order to prove the theorem, we use the following lemma:

\begin{lemma}\label{Lemma_T_Haft}
To ensure the feasibility of the power control problem in \eqref{robust-prob}, it is sufficient to have
\be\label{sufficient1}
\frac{\tan{\phi_k}}{\sin{\theta_k}}<\frac{1}{1+\sqrt{(M-1)\gamma_k}},
\ee
for all active users $k\in\mc{K}$.
\end{lemma}

\begin{proof}
The idea is to show that if the condition \eqref{sufficient1} holds, the upper bound solution in \eqref{Pub} is a valid solution to the power control problem in \eqref{robust-prob}. For this purpose, it suffices to show that $\sum_{k\in\mc{K}}{\alpha_k}{<}1$, where $\alpha_k$ is defined in \eqref{Pub}.

According to the condition \eqref{sufficient1}, we have \[\frac{\tan{\phi_k}}{\sin{\theta_k}}<\frac{1}{1+\sqrt{(M-1)\gamma_k}}.\] Then
\be\label{7-1}
\frac{\sin(\theta_k{-}\phi_k)}{\sin{\phi_k}} = \frac{\sin\theta_k}{\tan\phi_k}{-}\cos\theta_k>\frac{\sin\theta_k}{\tan\phi_k}{-}1>\sqrt{(M{-}1)\gamma_k}.
\ee
Therefore
\be\label{7-2}
\alpha_k=\left(1+\frac{\sin^2{(\theta_k-\phi_k)}}{\gamma_k\sin^2{\phi_k}}\right)^{-1}<\frac{1}{M},
\ee
and $\sum_{k\in\mc{K}}{\alpha_k}<1$, since $|\mc{K}|\leq M$.
\end{proof}

The MQCS condition in \eqref{MQCS} is equivalent to
\be\label{7-2.5}
4\lambda_M \NkD ^{-\frac{1}{M-1}} < \sin\left(\arctan\left(\frac{\sin{\theta_k^\circ}}{1+\sqrt{(M-1)\gamma_k}} \right)\right).
\ee
Combining this with the inequality in \eqref{phi-bound-mu}, we have
\be\label{7-2.7} \sin{\phi_k} < \sin\left(\arctan\left(\frac{\sin{\theta_k^\circ}}{1+\sqrt{(M-1)\gamma_k}} \right)\right),\ee
which leads to
\be\label{7-2.6}\tan{\phi_k} < \frac{\sin{\theta_k^\circ}}{1+\sqrt{(M-1)\gamma_k}}.\ee

Now for all active users $k\in\mc{K}$, we have $\theta_k\geq \theta_k^\circ$, since $\theta_k^\circ$ is the smallest acceptable angle between user $k$'s channel and the subspace spanned by all other users' channels. Combining $\theta_k\geq \theta_k^\circ$ with \eqref{7-2.6} leads to \[\tan{\phi_k} < \frac{\sin{\theta_k}}{1+\sqrt{(M-1)\gamma_k}},\]
which according to Lemma \ref{Lemma_T_Haft} guarantees feasibility and therefore completes the proof.

\subsection{Proof of Theorem \ref{T_Shish}} \label{A_2}

Consider the definition of the uncertainty region \[S_k=\left\{\left.\mbf{w}\in \ds{R}^M \right|\sqrt{r}\leq \|\mbf{w}\| < \sqrt{R}, ~ \angle(\mbf{w},\tilde{\mbf{u}})<\phi \right\}.\] As far as the constraint in \eqref{robust-const} in concerned, the constraint $\|\mbf{w}\| < \sqrt{R}$ is redundant since the infimum occurs on the lower surface $ \|\mbf{w}\|=\sqrt{r}$ (see the process of deriving the upper bound solution in \eqref{comp3}).

Now define
\begin{align}
&T_0=\left\{\bw\left|\frac{P_k\left|\mathbf{w}^{T}\mathbf{v}_k\right|^2}{\sum\limits_{l\neq
k}{P_l\left|\mathbf{w}^{T}\mathbf{v}_l\right|^2}+1}>\gamma_k\right.\right\},\\
&T_1=\left\{\bw\left|\|\mbf{w}\|\geq\sqrt{r}\right.\right\},\\
&T_2=\left\{\bw\left||\mbf{w}^T\tilde{\bu}_k|>\|\bw\|\cos\phi\right.\right\}, \label{T2}
\end{align}
where the condition in \eqref{T2} is equivalent to $\angle(\mbf{w},\tilde{\mbf{u}})<\phi$. With these definitions, the constraint in \eqref{robust-const} is equivalent to
\be\label{6-2} T_1\cap T_2 ~\subset~ T_0. \ee

Now, define scalars $\tau_n$ and symmetric matrices $A_n$, for $n=0,1,2$, as follows:
\begin{align}
&A_0={\sum_{l\neq k}{P_l\bv_l\bv_l^T}} - {\frac{1}{\gamma_k}{P_k\bv_k\bv_k^T}}\nonumber\\
&A_1=-\mathbf{I}\nonumber\\
&A_2 = -\frac{1}{\cos^2\phi}\tilde{\bu}_k\tilde{\bu}_k^T+\mbf{I}\nonumber\\
&\tau_0=-1,~~ \tau_1=-r, ~~ \tau_2= 0.\nonumber
\end{align}
With these definitions, the sets $T_0$, $T_1$, $T_2$  can be expressed as sublevels of quadratic functions: \[T_n=\left\{\bw\left| \bw^T A_n \bw \leq \tau_n \right.\right\}, ~~n=0,1,2.\]
We therefore can use the following theorem from \cite{polyak} to replace the condition \eqref{6-2} with a SDP condition. This theorem is an important extension of what is known as S-procedure in the optimization literature \cite{boyd}.

\begin{theorem}
Let $M\geq3$ and $A_n\in\ds{R}^{M\times M}$ be symmetric matrices for $n=0,1,2$ and assume \[\exists ~\nu_1,\nu_2\in \ds{R} ~~ \md{s.t.} ~~\nu_1A_1+\nu_2A_2\succ 0. \] Define the quadratic functions $f_n(\bw)=\bw^T A_n \bw$. Then the following two statements are equivalent:
\begin{align}
&\md{I.~~~~} f_1(\bw)\leq \tau_1~,~ f_2(\bw)\leq \tau_2 ~~ \Rightarrow ~~ f_0(\bw)\leq \tau_0 \\
&\md{II.~~~} \exists ~\lambda>0,~\mu>0 ~~~\md{s.t.}~~~ \left\{\begin{array}{lll} A_0 \preceq \lambda A_1 + \mu A_2 \\ \tau_0\geq \lambda \tau_1 + \mu \tau_2  \end{array} \right.
\end{align}
\end{theorem}

For the problem in hand, one can easily find $\nu_1$ and $\nu_2$ such that the condition  $\nu_1A_1+\nu_2A_2\succ 0$ is satisfied; therefore, the constraint in \eqref{6-2} translates to the SDP constraints in \eqref{SDP}.

\subsection{Proof of Theorem \ref{T_Hasht}} \label{A_3}

From the definitions of $\alpha_k$ and $\beta_k$ in \eqref{Pub}, we have
\begin{align}
&\alpha_k=\frac{\gamma_k\sin^2{\phi_k}}{\gamma_k\sin^2{\phi_k}+\sin^2{(\theta_k-\phi_k)}}, \nonumber\\
&\frac{\alpha_k}{\beta_k}=\frac{\gamma_k/r_k}{\gamma_k\sin^2{\phi_k}+\sin^2{(\theta_k-\phi_k)}}.\nonumber
\end{align}
For small values of $\phi_k\ll 1$, it is easy to verify that
\begin{align}
&\gamma_k\sin^2{\phi_k}=\gamma_k\phi_k^2+o(\phi_k^2),\nonumber\\
&\sin^2{(\theta_k-\phi_k)}=\sin^2\theta_k-(\sin 2\theta_k)\phi_k + o(\phi_k).\nonumber
\end{align}

After a few manipulations, one can show that
\begin{align} \label{8-3}
&\alpha_k= e_k r_k\phi_k^2 + o(\phi_k^2),\nonumber\\
&\frac{\alpha_k}{\beta_k}=e_k+f_k\phi_k+o(\phi_k),\nonumber
\end{align}
and
\be\label{8-40}
P_{_{\md{MU}}}=\frac{\sum_{k\in\mc{K}}{\alpha_k/\beta_k}}{1-\sum_{k\in\mc{K}}{\alpha_k}}=\sum_{k\in\mc{K}}{e_k{+}f_k\phi_k{+}o(\phi_k)}.\nonumber
\ee

%


\subsection{Proof of Theorem \ref{T_Noh}} \label{A_4}

By applying Lagrange multipliers method, we achieve
\begin{align}
&\NkM=\frac{1}{\Lambda}\cdot \frac{\gamma_k}{\theta_k^\circ} \label{9-2-1} \\
&\NkD=\frac{1}{\Lambda^{M-1}}\cdot \left(\frac{4\lambda_M}{M-1}\cdot \frac{\gamma_k}{(\theta_k^\circ)^2} \right)^{M-1}, \label{9-2-2}
\end{align}
where $\Lambda$ is a Lagrange multiplier that should satisfy the constraint $\prod_{k=1}^{M}{\NkM\NkD}=N$. By solving for $\Lambda$ we get
\be\label{9-3}
\Lambda=\left(\frac{4\lambda_M}{M-1}\right)^{\frac{M-1}{M}}\left(\prod_{k=1}^M \gamma_k\right)^{\frac{1}{M}}\left(\prod_{k=1}^M \frac{1}{\theta_k^\circ}\right)^{\frac{2M-1}{M^2}}\cdot N^{-\frac{1}{M^2}}.
\ee
By substituting \eqref{9-3} in \eqref{9-2-1} and \eqref{9-2-2}, and using $\theta_k^\circ=\frac{\pi}{4}q_k$ as in \eqref{theta0}, and further manipulation, the optimal quantization resolutions $\BkM=\log\NkM$ and $\BkD=\log\NkD$ can be expressed as in \eqref{OptimalBits1} and \eqref{OptimalBits2}.

\subsection{Proof of Theorem \ref{T_Dah}} \label{A_5}

For the target parameters to be feasible, the optimal direction codebook sizes $\NkD{=}2^{\BkD}$ are required to satisfy the MQCS conditions in \eqref{const3}. With the assumption of $q_k\ll 1$, we have $\sin\theta_k^\circ\approx \theta_k^\circ\ll1$ and the MQCS conditions simplify to the following conditions:
\be\label{NkD-lower-bound}
\NkD\geq \left(\frac{4\lambda_M}{\theta_k^\circ}\left(1+\sqrt{(M-1)\gamma_k}\ \right)\right)^{M-1}.
\ee
In the following, we find a lower bound on $B$ that guarantees the conditions in \eqref{NkD-lower-bound}.

For $M\geq 2$ and $\gamma_k>1$, we have the following inequality:
\be\label{10-1} 1+\sqrt{(M-1)\gamma_k}<\sqrt{2M\gamma_k}. \ee
To satisfy the conditions in \eqref{NkD-lower-bound}, therefore, it is sufficient to satisfy
\be\label{10-2}
\NkD\geq \left(\frac{4\lambda_M}{\theta_k^\circ}\sqrt{2M\gamma_k}\right)^{M-1}.
\ee
By substituting $\NkD$ from \eqref{9-2-2} in \eqref{10-2}, we achieve
\be\label{10-3} \frac{1}{\Lambda}\geq (M-1)\frac{\theta_k^\circ}{\gamma_k}\sqrt{2M\gamma_k}. \ee
Therefore it is sufficient to have
\be\label{10-4} \frac{1}{\Lambda}\geq \sqrt{2} M^{3/2} \frac{\theta_k^\circ}{\sqrt{\gamma_k}}. \ee

By substituting the expression for $\Lambda$ in \eqref{9-3} into \eqref{10-4}, we achieve the following constraints on $N=2^B$ for $1\leq k\leq M$:
\be\label{10-5}
N\geq C \cdot \left(\bar{\gamma}\cdot{\bar{\theta}}^{\circ^{-\frac{2M-1}{M}}} \cdot \frac{\theta_k^\circ}{\sqrt{\gamma_k}} \right)^{M^2}
\ee
where $\bar{\gamma}=\left(\prod_k{\gamma_k}\right)^{1/M}$ and $\bar{\theta}^\circ=\left(\prod_k{\theta_k^\circ}\right)^{1/M}$ and \[C=\left(\sqrt{2}M^{3/2}\right)^{M^2}\left(\frac{4\lambda_M}{M-1}\right)^{M(M-1)}.\]

Since \eqref{10-5} is to be satisfied for all $1\leq k\leq M$, we have the following sufficient bound on $N=2^B$:
\begin{align}\label{10-6}
&N\geq C \cdot \left(\bar{\gamma}\cdot{\bar{\theta}}^{\circ^{-\frac{2M-1}{M}}} \cdot \max_{1\leq k\leq M}\frac{\theta_k^\circ}{\sqrt{\gamma_k}} \right)^{M^2}\nonumber\\
&~~= C\cdot \left(\sqrt{\bar{\gamma}} \cdot {\bar{\theta}}^{\circ^{-\frac{M-1}{M}}}\cdot \max\limits_{1{\leq} k{\leq} M}\frac{{{\theta_k^\circ}/{\sqrt{\gamma_k}}}}{{{\bar{\theta}}^{\circ}}/{\sqrt{\gbar}}} \right)^{M^2}.
\end{align}
By substituting $\theta_k^\circ=\frac{\pi}{4}q_k$ and $\bar{\theta}^\circ=\frac{\pi}{4}\qbar$ in \eqref{10-6} and taking the logarithm of the both sides we achieve the lower bound in \eqref{total-rate}, which completes the proof.

\subsection{Proof of Theorem \ref{T_Yazdah}} \label{A_6}

By substituting the optimal values of $\NkM$ and $\NkD$ given by \eqref{9-2-1} and \eqref{9-2-2} into the average sum-power upper bound in \eqref{final-power-bound}, we have
\begin{align}\label{11-1}
\mathbb{E}\left[P_{_{\md{MU}}}\right]&< \frac{2\rho_{_{\md{MU,CSI}}}}{\pi} \sum_{k=1}^M{\frac{\gamma_k}{\theta_k^\circ}\left(1+\NkM^{-1}+\frac{4\lambda_M}{\theta_k^\circ}\NkD^{-\frac{1}{M-1}}\right)}\nonumber\\
&=\frac{2\rho_{_{\md{MU,CSI}}}}{\pi} \sum_{k=1}^M{\left(\frac{\gamma_k}{\theta_k^\circ}+\Lambda+(M-1)\Lambda\right)}\nonumber\\
&=\frac{2\rho_{_{\md{MU,CSI}}}}{\pi} \left({\left[\sum_{k=1}^M\frac{\gamma_k}{\theta_k^\circ}\right]+M^2\Lambda}\right)\nonumber\\
&\stackrel{\md{(a)}}{=}\frac{2\rho_{_{\md{MU,CSI}}}}{\pi} \left(\left[\sum_{k=1}^M\frac{\gamma_k}{\theta_k^\circ}\right]+M^2\chi\gbar\left(\frac{1}{\bar{\theta^\circ}}\right)^{2-\frac{1}{M}} 2^{-\frac{B}{M^2}}\right),
\end{align}
where $\bar{\gamma}{=}\left(\prod_k{\gamma_k}\right)^{1/M}$ and $\bar{\theta}^\circ{=}\left(\prod_k{\theta_k^\circ}\right)^{1/M}$. For the equality (a) we have used the expression \eqref{9-3} for $\Lambda$ with $N=2^B$ and $\chi=\left({4\lambda_M}/({M-1})\right)^{({M-1})/{M}}$.

By further manipulating \eqref{11-1}, we have
\begin{align}\label{11-2}
&\mathbb{E}\left[P_{_{\md{MU}}}\right]\nonumber\\
&<\frac{2\rho_{_{\md{MU,CSI}}}}{\pi} \left(\left[\sum_{k=1}^M\frac{\gamma_k}{\theta_k^\circ}\right]+M\chi\left[M\frac{\gbar}{\bar{\theta^\circ}}\right]
\left(\frac{1}{\bar{\theta}^\circ}\right)^{1-\frac{1}{M}} 2^{-\frac{B}{M^2}}\right)\nonumber\\
&\stackrel{\md{(b)}}{<}\frac{2\rho_{_{\md{MU,CSI}}}}{\pi} \left(\left[\sum_{k=1}^M\frac{\gamma_k}{\theta_k^\circ}\right]+M\chi\left[\sum_{k=1}^M\frac{\gamma_k}{\theta_k^\circ}\right]
\left(\frac{1}{\bar{\theta}^\circ}\right)^{1-\frac{1}{M}} 2^{-\frac{B}{M^2}}\right)\nonumber\\
&= \frac{2\rho_{_{\md{MU,CSI}}}}{\pi} \left[\sum_{k=1}^M\frac{\gamma_k}{\theta_k^\circ}\right] \left(1+\frac{M\chi}{\bar{\theta}^{\circ}}{\bar{\theta}^{\circ^{\frac{1}{M}}}}2^{-\frac{B}{M^2}}\right)\nonumber\\
&\stackrel{\md{(c)}}{<}P_{_{\md{MU,CSI}}} \left(1+\frac{M\chi}{\bar{\theta}^{\circ}}2^{-\frac{B}{M^2}}\right),
\end{align}
where in (b) we use the fact the geometric mean is smaller than arithmetic mean:
\[\frac{\gbar}{\bar{\theta}^\circ}=\left(\prod_{k=1}^M{\frac{\gamma_k}{\theta_k}}\right)^{\frac{1}{M}}\leq \frac{1}{M}\sum_{k=1}^M{\frac{\gamma_k}{\theta_k}},\]
and in deriving (c), we use the definition of $P_{_{\md{MU,CSI}}}$ in \eqref{P-CSI} and the fact that $\bar{\theta}^\circ=\frac{\pi}{4}\qbar<1$, since $\qbar$ is a probability measure.

By substituting $\bar{\theta}^\circ=\frac{\pi}{4}\qbar$ in \eqref{11-2}, we achieve the bound in \eqref{power-scaling} and the proof is complete.

\end{document}